\documentclass[a4paper,10pt]{scrartcl}

\usepackage[utf8]{inputenc}
\usepackage[T1]{fontenc}
\usepackage[english]{babel}

\usepackage{amsmath,amssymb,amsfonts,amsthm}

\usepackage[ruled, vlined, linesnumbered]{algorithm2e}
\usepackage{xspace}

\usepackage{fancyhdr}

\usepackage[margin=3cm]{geometry}
\newcommand{\R}{\mathbb{R}}

\newcommand{\N}{\mathbb{N}}

\newcommand{\cH}{\mathcal{H}}

\newcommand{\abs}[1]{\left\vert #1 \right\vert}
\newcommand{\norm}[1]{\left\Vert #1 \right\Vert}
\newcommand{\sca}[2]{\langle #1 , #2 \rangle}
\newcommand{\OneTo}[2]{#1\in[#2]}

\newcommand{\wmax}{w_{max}^B}
\newcommand{\wmin}{w_{min}^B}
\newcommand{\cif}{& \text{if }}

\newcommand{\SAT}{\textsc{Max 2-Sat}\xspace}
\newcommand{\SATF}{{\normalfont SAT/Flip}\xspace}
\newcommand{\KM}{\textsc{Discrete $K$-Means}\xspace}
\newcommand{\KMS}{{\normalfont DKM/Swap}\xspace}
\newcommand{\UFL}{\textsc{Uncapacitated Facility Location}\xspace}
\newcommand{\MUFL}{\textsc{Metric Uncapacitated Facility Location}\xspace}
\newcommand{\MUFLS}{{\normalfont MUFL/Swap}\xspace}

\newcommand{\pkm}{\phi_{KM}}
\newcommand{\pfl}{\phi_{FL}}

\newtheorem{definition}{Definition}

\newtheorem{theorem}[definition]{Theorem}
\newtheorem{lemma}[definition]{Lemma}

\newtheorem{corollary}[definition]{Corollary}

\newtheorem{proposition}[definition]{Proposition}

\setkomafont{paragraph}{\normalfont\itshape\large}

\fancypagestyle{plain}{
\fancyhf{}
\cfoot{\normalfont \hfill \pagemark \hfill}

}

\allowdisplaybreaks

\makeatletter
\newcommand*{\defeq}{\mathrel{\rlap{%
                     \raisebox{0.3ex}{$\m@th\cdot$}}%
                     \raisebox{-0.3ex}{$\m@th\cdot$}}%
                     =}
\makeatother

\begin{document}

\thispagestyle{empty}

\begin{center}
	{\LARGE \textbf{A Technical Report on PLS-Completeness of Single-Swap for Unweighted Metric Facility Location and $K$-Means}}
	\bigskip
	
	{\Large Sascha Brauer}
	
	\texttt{sascha.brauer{@}uni-paderborn.de}
	\bigskip
	
	Department of Computer Science
	
	Paderborn University
	
	33098 Paderborn, Germany
\end{center}  
\begin{abstract}
	Recently, \cite{brauer17} showed that the single-swap heuristic for weighted metric uncapacitated facility location and $K$-Means is tightly PLS-complete.
	We build upon this work and present a stronger reduction, which proves tight PLS-completeness for the unweighted version of both problems.
\end{abstract}

\section{Introduction}

Metric facility location and $K$-Means are important problems in operations research and computational geometry.
Both problems admit a fairly simple local search scheme called \emph{single-swap}, which is known to compute a constant factor approximation of an optimal solution.
Recently, \cite{brauer17} showed that single-swap is \emph{tightly PLS-complete}, which means that the local search algorithm requires exponentially many steps in the worst case and that given some initial solution it is PSPACE-complete to find the solution computed by the algorithm started on this initial solution.
One shortcoming of the presented result is that it constructs a non-trivial weight function on the clients, hence, the reduction is not sufficient to classify the hardness of the problems on unweighted instances.

In this report, we modify the reduction presented in \cite{brauer17} to obtain tight PLS-completeness of unweighted versions of metric uncapacitated facility location and discrete $K$-Means.
This is an important extension of the result, since, besides being a formally stronger result, unweighted instances are significantly more relevant in practice.
Furthermore, we present a lower bound on the number of dimensions required to embed the point set constructed by our reduction into Euclidean space.
\section{Preliminaries and Notation}

In an \UFL (UFL) problem we are given a set of clients $C$, a set of facilities $F$, an opening cost function $f:F\rightarrow\R$, and a distance function $d: C\times F\rightarrow \R$.
The goal is to find a subset of facilities $O\subset F$ minimizing

\[ \pfl(C,F,O) = \sum_{c\in C} \min_{o\in O}\{d(c,o)\} + \sum_{o\in O} f(o) \ . \]

\MUFL is a special case of this problem, where we require the distance function $d$ to be a metric on $C\cup F$.
The PLS problem \MUFLS consists of MUFL, where the so-called \emph{single-swap} neighbourhood of a set of open facilities is given by sets of facilities obtained by newly opening a closed facility, closing an open facility, or doing both in one step (\emph{swapping} two facilities).

In \KM (DKM) we do not differentiate between clients and facilities, but are given a single set of points $C\subset\R^D$.
We measure distance between points $p,q\in C$ as $d(p,q) = \norm{p-q}^2$.
Furthermore, instead of imposing an opening cost, we allow at most $K$ locations to be opened.
Hence, the goal is to find $O\subset C$ with $\abs{O} = K$ minimizing

\[ \pkm(C,O) = \sum_{c\in C} \min_{o\in O}\{\norm{c - o}^2 \} \ . \]

The PLS problem \KMS consists of DKM, where the single-swap neighbourhood is given by all subsets of points of $C$ obtained by swapping two points.
We forbid the open and close operations, as these leave the space of feasible solutions for DKM.

\SAT (SAT) is a variant of the satisfiability problem, where each clause consists of exactly $2$ literals and has some positive integer weight assigned to it.
The cost of a truth assignment is the sum of the weights of all satisfied clauses.
The PLS problem \SATF consists of SAT, where the neighbourhood of an assignment is given by all assignments obtained by changing the truth value of a single variable.

\begin{theorem}[\cite{schaeffer91}]
	\SATF is tightly PLS-complete.
\end{theorem}

For each clause set $B$ and truth assignment $T$ we denote the SAT cost of $T$ with respect to $B$ by $w(B,T)$.
For a literal $x$ we denote the set of all clauses in $B$ containing $x$ by $B(x)$.
Further, we denote the set of all clauses in $B$ satisfied by $T$ by $B_t(T)$ and let $B_f(T) = B\setminus B_t(T)$.
Finally, we set $\wmax = \max_{b\in B}\{ w(b) \}$ and $\wmin = \min_{b\in B}\{ w(b)\}$.
\section{Facility Location}

We show that unweighted \MUFLS is tightly PLS-complete.

\begin{proposition}\label{prop:mufl}
	\SATF $\leq_{PLS}$ \MUFLS and this reduction is tight.
\end{proposition}

In the following, we present our PLS-reduction $(\Phi,\Psi)$ and prove its correctness.
We omit the tightness proof, as is can easily by obtained by incorporating the arguments presented here into the tightness proof from \cite{brauer17}.

\subsection{Construction of $\Phi$ and $\Psi$}

First, we construct the function $\Phi$ mapping a \SAT instance to a \MUFL instance.
Let $(B,w)$ be a \SAT instance over the variables $\{x_n\}_{\OneTo{n}{N}}$.
Let further $M\defeq \abs{B}$ and $W \defeq M\cdot\wmax$.
We assume that $M \geq 2$.
In the following, we construct an instance $(C,F,f,d)\in$ \MUFL.
We set $F = \{x_n,\bar x_n\}_{\OneTo{n}{N}}$ and locate a client at each facility and a client corresponding to each clause, so $C = F\cup B$.
We set $d:C\cup F \times C\cup F \rightarrow \R$ to

\[ d(p,q) = d(q,p) = 
	\begin{cases}
		0 \cif p = q \\
		1 \cif p = x_n \wedge q = \bar x_n \\
		1+\frac{w(b_m)}{W} \cif (p = x_n \vee p = \bar x_n) \wedge q = b_m \wedge p\in b_m \\
		1+\frac{c\cdot w(b_m)}{W} \cif (p = x_n \vee p = \bar x_n) \wedge q = b_m \wedge \bar p\in b_m \\
		2 & \text{else,}
	\end{cases}
\]
where $1 < c < 2$.

It is easy to see that $d$ is a metric and that the points closest to each other are literals and their negation, that clauses are closer to literals the contain, then to the literal's negation, and that all other point pairs have the same, even larger, distance to each other.

The opening cost function is constant $f \equiv 2$.

Second, we construct the function $\Psi$ mapping solutions of $\Phi(B,w)$ back to solutions of $(B,w)$.
Given a set $O\subset F$ we let each variable $x_n$ be true if $x_n\in F$ and let it be false otherwise.

In the following, we denote $\Phi(B,w) = (C,F,2,d)$ and $\Psi(B,w,O) = T_O$.

\subsection{$(\Phi,\Psi)$ is a PLS-reduction}

We need to argue that $T_O$ is locally optimal for $(B,w)$ if $O$ is locally optimal for $\Phi(B,w)$.
To prove this, we define a particular subset of solutions for $\Phi(B,w)$ we call \emph{reasonable} solutions.

\begin{definition}
	Let $O\subset F$.
	We call $O$ \emph{reasonable} if $\abs{O} = N$ and
	\[ \forall \OneTo{n}{N}: x_n\in O \vee \bar x_n\in O \ . \]
\end{definition}

To prove correctness of our reduction, we observe the following crucial properties of reasonable solutions.
The restriction of $\Psi$ to reasonable solutions is a bijection, there is close relation of the MUFL cost of a reasonable solution and the cost of its image under $\Psi$, and all locally optimal solutions of $(C,F,2,d)$ are reasonable.

\begin{lemma}\label{lem:uflcost}
	If $O\subset F$ is reasonable, then
	\[ \pfl(C,F,O) = 3N + M + \frac{1}{W} \sum_{b_m\in B} w(b_m) + \frac{c-1}{W}\sum_{b_m\in B_f(T_O)} w(b_m) \ . \]
\end{lemma}

\begin{proof}
	Observe, that reasonable solutions incur opening cost of $2N$.
	By definition of $\Psi$ and reasonable solutions, we have that a variable $x_n$ is assigned \emph{true} if and only if $x_n\in O$ and is assigned \emph{false} if and only if $\bar x_n \in O$.
	We obtain that the service cost of clients corresponding to literals is $N$, since each facility corresponding to a literal either is in $O$, then the corresponding client has cost $0$, or its negated facility at distance $1$ is in $O$ and it has cost $1$.
	It is easy to see, by definition of the point set and $\Psi$, that a client corresponding to a clause $b_m\in B_t(T_O)$ has at least one facility at distance $1+ w(b_m)/W$ and that a client corresponding to a clause $b_m\in B_f(T_O)$ has two facilities at distance $1+c\cdot w(b_m)/W$ and the rest at distance $2$.
	We obtain
	\begin{align*}
		\pfl(C,F,O) &= 2N + N + \sum_{b_m\in B_t(T_O)} \left( 1 + \frac{w(b_m)}{W}\right) + \sum_{b_m\in B_f(T_O)} \left( 1 + \frac{c\cdot w(b_m)}{W}\right) \\
		&= 3N + M + \sum_{b_m\in B_t(T_O)} \frac{w(b_m)}{W} + \sum_{b_m\in B_f(T_O)} \frac{c\cdot w(b_m)}{W} \\
		&= 3N + M + \frac{1}{W}\sum_{b_m\in B} w(b_m) + \frac{c-1}{W}\sum_{b_m\in B_f(T_O)} w(b_m)
	\end{align*}

\end{proof}

\begin{corollary}\label{cor:flcost}
	If $O,O'\subset C$ are reasonable solutions for $\Phi(B,w)$, then
	\[ w(B,T_O) < w(B, T_{O'}) \Leftrightarrow \pfl(C,F,O) > \pfl(C,F,O') \ . \]
\end{corollary}

\begin{proof}
	Observe that 
	\begin{align*} 
	w(B,T_O) = &\sum_{b_m\in B_t(T_O)} w(b_m) < \sum_{b_m\in B_t(T_{O'})} w(b_m) = w(B, T_{O'}) \\
	\Leftrightarrow &\sum_{b_m\in B_f(T_O)} w(b_m) > \sum_{b_m\in B_f(T_{O'})} w(b_m) \ . 
	\end{align*}
	Observe, that, by Lemma~\ref{lem:uflcost}, the only summand of the cost $\pfl(C,F,O)$ of a reasonable solution $O$ actually depending on $O$ is $(c-1)/W\sum_{b_m\in B_f(T_O)} w(b_m)$.
	Since $c > 1$, we obtain
	\begin{align*}
		\sum_{b_m\in B_f(T_O)} w(b_m) &> \sum_{b_m\in B_f(T_{O'})} w(b_m) \\
		\Leftrightarrow \pfl(C,F,O) &> \pfl(C,F,O') \ . 
	\end{align*}
\end{proof}

\begin{lemma}\label{lem:fllocalopt}
	If $O\subset F$ is locally optimal for $\Phi(B,W)$, then $O$ is reasonable.
\end{lemma}

\begin{proof}
	The following proof is presented in two steps.
	First, we argue that no locally optimal solution can contain both a literal and its negation.
	Second, we show that every locally optimal solution contains a facility corresponding to each of the variables.
	
	Assume, there is an $n$ such that $x_n, \bar x_n\in O$.
	Observe, that $B(x_n)\leq M$ and that no client in $C\setminus (B(x_n)\cup\{x_n\})$ is closer to $x_n$ than it is to $\bar x_n$.
	Since $1 < c < 2$ and by definition of $W$, we obtain
	\begin{align*}
		\pfl(C,F,O) &= \sum_{c\in C\setminus (B(x_n)\cup\{x_n\})} d(c,O) + \sum_{b_m\in B(x_n)} \left( 1 + \frac{w(b_m)}{W}\right) + \abs{O}2 \\
		&>\sum_{c\in C\setminus (B(x_n)\cup\{x_n\})} d(c,O) + \sum_{b_m\in B(x_n)} \left( 1 + \frac{c\cdot w(b_m)}{W}\right) + \underbrace{1}_{d(x_n,\bar x_n)} + (\abs{O}-1)2 \\
		&= \pfl(C,F,O\setminus \{x_n\})
	\end{align*}

	Now assume, that $x_n, \bar x_n \not\in O$.
	Connecting the clients located at $x_n$ and $\bar x_n$ to a facility newly opened at $x_n$ is sufficient to reduce the overall cost.
	We obtain
	\begin{align*}
		\pfl(C,F,O) &= \sum_{c\in C\setminus\{x_n,\bar x_n\}} d(c,O) + \underbrace{\sum_{c\in \{x_n,\bar x_n\}} d(c,O)}_{ = 4} + \abs{O}2 \\
		&> \sum_{c\in C\setminus\{x_n,\bar x_n\}} d(c,O) + \underbrace{1}_{d(x_n,\bar x_n)} + (\abs{O}+1)2 \\
		&\geq \pfl(C,F,O\cup\{x_n\})
	\end{align*}
\end{proof}

\begin{corollary}
	If $O$ is locally optimal for $\Psi(B,w)$, then $T_O$ is locally optimal for $(B,w)$.
\end{corollary}

\begin{proof}
	Combine Corollary~\ref{cor:flcost} and Lemma~\ref{lem:fllocalopt}.
\end{proof}

\section{$K$-Means}

We complement our results by showing that we can obtain tight PLS-complete\-ness for \KMS, as well.

\begin{proposition}\label{prop:km}
	\SATF $\leq_{PLS}$ \KMS and this reduction is tight.
\end{proposition}

This reduction is similar to the previously presented reduction for MUFL. 
We mostly have to change some of the constants involved, and finally argue that there is a point set in $\R^D$ exhibiting the required interpoint distances.

\subsection{Constructing $\Phi$ and $\Psi$}

First, we construct the function $\Phi$ mapping a \SAT instance to a \KM instance.
Let $(B,w)$ be a \SAT instance over the variables $\{x_n\}_{\OneTo{n}{N}}$.
Let further $M\defeq \abs{B}$ and $W \defeq M\cdot\wmax$.
We assume that $M \geq 2$.
In the following, we construct an instance $(C,K)\in$ \KM.
Abstractly define the point set $C = \{x_n,\bar x_n\}_{\OneTo{n}{N}}\cup B$.
The distance function $d:C\times C\rightarrow \R$ is defined as
\[ d(p,q) = d(q,p) = 
	\begin{cases}
		0 \cif p = q \\
		1 \cif p = x_n \wedge q = \bar x_n \\
		1+\epsilon \left(\frac{3}{2}+\frac{w(b_m)}{2W}\right)\cif (p = x_n \vee p = \bar x_n) \wedge q = b_m \wedge p\in b_m \\
		1+\epsilon \left(\frac{3}{2}+\frac{c\cdot w(b_m)}{2W}\right) \cif (p = x_n \vee p = \bar x_n) \wedge q = b_m \wedge \bar p\in b_m \\
		1+2\epsilon & \text{else,}
	\end{cases}
\]
where $1 < c < 2$ and 
\[ \epsilon = \frac{1}{4N+2M} \ . \]

One can easily see that the points closest to each other are literals and their negation, that clauses are closer to literals the contain, then to the literal's negation, and the all other point pairs have the same, even larger, distance to each other.

Finally, we choose $K = N$.

Second, we construct the function $\Psi$ mapping solutions of $\Phi(B,w)$ back to solutions of $(B,w)$.
Given a set $O\subset C$ we let each variable $x_n$ be true if $x_n\in O$ and let it be false otherwise.

In the following, we denote $\Phi(B,w) = (C,N)$ and $\Psi(B,w,O) = T_O$.

\subsection{$(\Phi,\Psi)$ is a PLS-reduction}

Recall the definition of reasonable solutions and observe that all reasonable solutions are also feasible solutions for the DKM instance $(C,N)$.

\begin{lemma}\label{lem:dkmcost}
	If $O\subset C$ is reasonable, then
	\[ \pkm(C,O) = N + M\left(1+\epsilon\frac{3}{2}\right) + \frac{\epsilon}{2W}\sum_{b_m\in B} w(b_m) + \frac{\epsilon}{2W}(c-1)\sum_{b_m \in B_f(T_O)} w(b_m) \ . \]
\end{lemma}

\begin{proof}
	By definition of $\Psi$ and reasonable solutions, we have that a variable $x_n$ is assigned \emph{true} if and only if $x_n\in O$ and is assigned \emph{false} if and only if $\bar x_n \in O$.
	We obtain that $\pkm(\{x_n,\bar x_n\}_{\OneTo{n}{N}},O) = N$, since each point corresponding to a literal is either in $O$ and has cost $0$, or its negated literal at distance $1$ is in $O$ and it has cost $1$.
	It is easy to see, by definition of the point set and $\Psi$, that a point corresponding to a clause $b_m\in B_t(T_O)$ has at least one mean at distance $1+\epsilon (3/2 + w(b_m)/(2W))$ and that a point corresponding to a clause $b_m\in B_f(T_O)$ has two means at distance $1+\epsilon (3/2 + c\cdot w(b_m)/(2W))$ and the rest at distance $1+2\epsilon$.
	We obtain
	\begin{align*}
		\pkm(C,O) &= \pkm(\{x_n,\bar x_n\}_{\OneTo{n}{N}},O) + \pkm(B,O) \\
		&= N + \sum_{b_m\in B_t(T_O)} \left( 1+\epsilon \left(\frac{3}{2}+\frac{w(b_m)}{2W}\right) \right) + \sum_{b_m\in B_f(T_O)} \left( 1+\epsilon \left(\frac{3}{2}+\frac{c\cdot w(b_m)}{2W}\right) \right) \\
		&= N + M\left(1+\epsilon\frac{3}{2}\right) + \sum_{b_m\in B_t(T_O)} \epsilon \frac{w(b_m)}{2W} + \sum_{b_m\in B_f(T_O)} \epsilon \frac{c\cdot w(b_m)}{2W} \\
		&= N + M\left(1+\epsilon\frac{3}{2}\right) + \frac{\epsilon}{2W}\sum_{b_m\in B} w(b_m) + \frac{\epsilon}{2W}(c-1)\sum_{b_m \in B_f(T_O)} w(b_m)
	\end{align*}
\end{proof}

\begin{corollary}
	If $O,O'\subset C$ are reasonable solutions for $\Phi(B,w)$, then
	\[ w(B,T_O) < w(B, T_{O'}) \Leftrightarrow \pkm(C,O) > \pkm(C,O') \ . \]
\end{corollary}

\begin{proof}
	Analogous to the proof of Corollary~\ref{cor:flcost}
\end{proof}

\begin{lemma}
	If $O\subset C$ is locally optimal for $\Phi(B,w)$, then $O$ is reasonable.
\end{lemma}

\begin{proof}
	Recall, that each point $b_m\in C$ has exactly two points at distance $1+\epsilon (3/2 + w(b_m)/(2W))$ and two points at distance $1+\epsilon (3/2 + c\cdot w(b_m)/(2W))$ (the points corresponding to the literals in the clause $b_m$ and their negations, respectively).
	In the following, we call these four points \emph{adjacent} to $b_m$.
	All the other points have distance $1+2\epsilon$ to $b_m$ and are hence strictly farther away.
	Assume to the contrary that there exists an $\OneTo{n}{N}$, such that $x_n,\bar x_n\not\in O$.

	\paragraph{Case 1:}
	There exists an $\OneTo{m}{M}: b_m \in O$, such that $b_m = \{x_o,x_p\}$ (where one or both of these literals might be negated).
	One important observation is that if we exchange $b_m$ for some other location then only its own cost and the cost of its adjacent points can increase.
	All other points, which might be connected to $b_m$, are at distance $1+2\epsilon$ and can hence be connected to any other location for at most the same cost.

	\paragraph{Case 1.1:}
	$x_o,\bar x_o, x_p, \bar x_p \not\in  O$.
	Each point adjacent to $b_m$ has distance at least $1+\epsilon (3/2 + \wmin/(2W))$ to every other points in $P$.
	Hence, we have that $\pkm( \{b_m,x_o,\bar x_o,x_p,\bar x_p\}, O) \geq 4+4\epsilon (3/2 + \wmin/(2W)) > 4+6\epsilon$.
	However,
	\[ \phi(\{b_m,x_o,\bar x_o, x_p, \bar x_p\},\{x_o\}) = 1+\epsilon (\underbrace{3/2 + w(b_m)/(2W)}_{<2}) + 1 + 2 + 4\epsilon < 4+6\epsilon \ , \]
	and hence $(O\setminus\{b_m\})\cup\{x_o\}$ is in the neighbourhood of $O$ and has strictly smaller cost.

	\paragraph{Case 1.2:}
	$x_p\in O \vee \bar x_p\in O$ and $x_o,\bar x_o\not\in O$.
	In this case, removing $b_m$ from $O$ does not affect the cost of $x_p$ and $\bar x_p$.
	We obtain $\phi(\{b_m,x_o,\bar x_o\},O) \geq 2+2\epsilon (3/2 + \wmin/(2W)) > 2+3\epsilon$.
	Observe, that
	\[ \phi(\{b_m,x_o,\bar x_o\},(O\setminus\{b_m\})\cup\{x_o\}) \leq 1+\epsilon(3/2 + w(b_m)/(2W)) + 1 < 2 + 2\epsilon \ . \]

	\paragraph{Case 1.3:}
	$x_p\in O \vee \bar x_p\in O$ and $x_o\in O \vee \bar x_o\in O$.
	Here we have that removing $b_m$ from $O$ does not affect the cost of its adjacent points at all.
	However, similar to before we have $\phi(\{b_m,x_n,\bar x_n\},O) \geq 2+2\epsilon (3/2 + \wmin/(2W)) > 2+3\epsilon$.
	Again, we obtain
	\[ \phi(\{b_m,x_n,\bar x_n\},(C\setminus\{b_m\})\cup\{x_n\}) \leq 1+\epsilon(3/2 + c\cdot w(b_m)/(2W)) + 1 < 2 + 2\epsilon \ . \]

	\paragraph{Case 2:}
	There is no $\OneTo{m}{M}$, such that $b_m\in O$.
	Consequently, there is an $\OneTo{o}{N}, o\neq n: x_o,\bar x_o\in O$.
	W.l.o.g. assume that $\abs{B(x_o)} < M$ (otherwise just exchange $x_o$ for $\bar x_o$ in the following argument).
	Observe that
	\[ \phi(B(x_o)\cup\{x_o,x_n,\bar x_n\},O) = 2 + 4\epsilon + \sum_{b_m\in B(x_o)} 1 + \epsilon (3/2 + w(b_m)/(2W))  \ . \]
	The only points affected by removing $x_o$ from $O$ are $x_o$ and the points corresponding to clauses in $B(x_o)$.
	Hence, 
	\begin{align*}
	\phi(C\setminus(B(x_o)\cup\{x_o,x_n,\bar x_n\}),O) &= \phi(C\setminus(B(x_o)\cup\{x_o,x_n,\bar x_n\}),O\setminus\{x_o\}) \\
	&\geq \phi(C\setminus(B(x_o)\cup\{x_o,x_n,\bar x_n\}),(O\setminus\{x_o\})\cup\{x_n\}) \  .
	\end{align*}
	However, recall that the points in $B(x_o)$ are at distance $1 + \epsilon (3/2 + c\cdot w(b_m)/(2W))$ from $\bar x_o\in O$.
	We obtain
	\begin{align*}
		&\phi(B(x_o)\cup\{x_o,x_n,\bar x_n\},(O\setminus\{x_o\})\cup\{x_n\}) \\
		&\leq \phi(B(x_o)\cup\{x_o,\bar x_n\},\{\bar x_o,x_n\}) \\
		&= 2 + \sum_{b_m\in B(x_o)} 1 + \epsilon (3/2 + c\cdot w(b_m)/(2W)) \\
		&= 2 + \sum_{b_m\in B(x_o)} 1 + \epsilon (3/2 + w(b_m)/(2W)) + (c-1)\epsilon/(2W) \underbrace{\sum_{b_m\in B(x_o)} w(b_m)}_{< M\wmax} \\
		&< 2 + \epsilon + \sum_{b_m\in B(x_o)} 1 + \epsilon (3/2 + w(b_m)/(2W)) \\
		&< \phi(B(x_o)\cup\{x_o,x_n,\bar x_n\},O) \ .
	\end{align*}

\end{proof}

\subsection{Embedding $C$ into $\ell_2^2$}\label{subsec:embedding}

So far, we regarded $C$ as an abstract point set, only given by fixed pairwise interpoint distances.
Using the results presented in \cite{brauer17} one can easily see that point set $C$ presented here can be embedded isometrically into squared Euclidean space, as well.
We complement this by showing a lower bound on the number of dimensions required to embed $C$ which asymptotically matches the number of dimensions of a possible embedding.

\begin{theorem}\label{thm:equidistances}
	If $X\subset\R^D$ is a set of pairwise equidistant points, then $\abs{X} \leq D+1$.
\end{theorem}

\begin{proof}
	By prove this theorem by induction on the number of dimensions.
	The claim obviously holds for $D=1$ (there are at most $2$ pairwise equidistant points on a line).
	Assume that the claim holds for any fixed $D\in\N$.
	Let $c\in\R_{>0}$ be a constant, $X \subset\R^{D+1}$ be a set of $D+1$ equidistant points and let $\cH$ be the $D$-dimensional hyperplane spanned by the points in $X$.
	We want to find another point $x_{D+2}\in\R^{D+1}$, such that $\forall x\in X: \norm{x-x_{D+2}}^2 = c$.
	By induction hypotheses, $\cH$ can not contain $x_{D+2}$.
	However, the mean $\mu(X) = \sum_{d=1}^{D+1} x_d / (D+1)$ is equidistant to all points in $X$.
	To see this, recall that $\norm{x-y}^2 = \norm{x}^2 + \norm{y}^2 - 2\sca{x}{y}$.
	For all $i\neq j$, we obtain $\sca{x_i}{x_j} = (\norm{x_i}^2 + \norm{x_j}^2 - c)/2$.
	Hence, for any fixed $x\in X$
	\begin{align*}
		\norm{x-\mu(X)}^2 &= \sca{x-\frac{1}{D+1}\sum_i x_i}{x-\frac{1}{D+1}\sum_j x_j} \\
		&= \frac{1}{(D+1)^2} \sum_i \sum_j \sca{x-x_i}{x-x_j} \\
		&= \frac{1}{(D+1)^2} \sum_{\substack{i \\ x_i\neq x}} \sum_{\substack{j \\ x_j\neq x}}\sca{x-x_i}{x-x_j} \\
		&= \frac{1}{(D+1)^2} \sum_{\substack{i \\ x_i\neq x}}\norm{x-x_i}^2  + \frac{1}{(D+1)^2}\sum_{\substack{i \\ x_i\neq x}} \sum_{\substack{j \neq i \\ x_j\neq x}}\sca{x-x_i}{x-x_j} \\
		&= \frac{cD}{(D+1)^2} + \frac{1}{(D+1)^2}\sum_{\substack{i \\ x_i\neq x}} \sum_{\substack{j \neq i \\ x_j\neq x}}\norm{x}^2 - \sca{x}{x_i} - \sca{x}{x_j} + \sca{x_i}{x_j} \\
		&= \frac{cD}{(D+1)^2} + \frac{1}{(D+1)^2}\sum_{\substack{i \\ x_i\neq x}} \sum_{\substack{j \neq i \\ x_j\neq x}}\frac{c}{2} =  \frac{cD}{(D+1)^2} \left(1 + \frac{D-1}{2}\right)\\
	\end{align*}
	We see that the points in $X$ all lie on the surface of a ball and that therefore the only point in $\cH$ having the same distance to all points in $X$ is $\mu(X)$.
	
	By definition of squared Euclidean distance, we know that the orthogonal projection $\pi_\perp^\cH(x_{D+2})$ of $x_{D+2}$ onto $\cH$ has to be equidistant to all points in $X$.
	From this, we obtain that $x_{D+2}$ has to lie on a line orthogonal to $\cH$ passing through $\mu(X)$.
	There are exactly two points on this line having distance $c$ to all points in $X$ (one on each side of $\cH$).
	Since the distance between these two points is larger than $c$, we can only add one of them to $X$ and retain pairwise equidistances.
\end{proof}

Observe, that the bound in Theorem~\ref{thm:equidistances} is sharp, since the set 
\[ \{(1,0,\dots,0),(0,1,0,\dots,0),\dots,(0,\dots,0,1), (\frac{\sqrt{D+1}+1}{D},\dots,\frac{\sqrt{D+1}+1}{D})\}\subset\R^D \]
is a set of $D+1$ points having pairwise distance $2$.

\begin{corollary}
	There is no isometric embedding of $C$ into squared Euclidean space using less then $\max\{N,M\}-1$ dimensions.
\end{corollary}

\begin{proof}
	By definition of our distances, the sets $B$ and $\{x_n\}_{\OneTo{n}{N}}$ form sets of pairwise equidistant points.
	Thus, we obtain the claim by applying Theorem~\ref{thm:equidistances}.
\end{proof}
\section{Open Problems}

By proving tight PLS-completeness for unweighted MUFL and DKM we solved one of the open problems posed by \cite{brauer17}.
We furthermore answered the question if it was possible to asymptotically reduce the dimensionality of the DKM reduction by presenting a lower bound on the number of dimensions required to embed the point set constructed by our reduction.
Since we asymptotically match the number of dimensions used when actually embedding the point set, this shows that it is not possible the obtain a better result using the reduction presented here.
Hence, to prove PLS-completeness of DKM in low dimensional space, we need to come up with a structurally new idea.

\bibliographystyle{alpha}
\bibliography{ref}

\end{document}